\documentclass[11pt]{article}
\usepackage{amsmath,amssymb,amsfonts,amsthm,epsfig}
\usepackage[usenames,dvipsnames]{xcolor}
\usepackage{bm,xspace}
\usepackage{tcolorbox}
\usepackage{cancel}
\usepackage{fullpage}
\usepackage{liyang}
\usepackage{framed}
\usepackage{verbatim}
\usepackage{enumitem}
\usepackage{array}
\usepackage{multirow}
\usepackage{afterpage}
\usepackage{mathrsfs}
\usepackage{pifont} 
\usepackage{chngpage}
\usepackage[normalem]{ulem}
\usepackage{boxedminipage}
\usepackage{caption}
\usepackage{forest}

\usepackage{algorithm}
\usepackage{algorithmicx}
\usepackage{algpseudocode}

\usepackage{pgfplots}
\pgfplotsset{width=8cm,compat=newest}

\usepackage{tikz}
\usetikzlibrary{calc,through,backgrounds}
\usepackage{tikz-cd}

\def\colorful{0}

\ifnum\colorful=1
\newcommand{\violet}[1]{{\color{violet}{#1}}}

\newcommand{\gray}[1]{{\color{gray}{#1}}}

\fi
\ifnum\colorful=0
\newcommand{\violet}[1]{{{#1}}}

\newcommand{\gray}[1]{{{}}}

\fi

\newcommand{\Sens}{\mathrm{Sens}}

\newcommand{\maxinf}{\mathrm{MaxInf}}
\newcommand{\balinf}{\textsc{BalInf}}
\newcommand{\bal}{\mathrm{bal}}

\def\mmww{\textsc{MMWW}}
\def\gmmww{\textsc{Gapped-MMWW}}

\newlist{enumprop}{enumerate}{1} 
\setlist[enumprop]{label=\arabic*.,ref=\theproposition.\arabic*}
\crefalias{enumpropi}{proposition}

\makeatletter
\newtheorem*{rep@theorem}{\rep@title}
\newcommand{\newreptheorem}[2]{
\newenvironment{rep#1}[1]{
 \def\rep@title{#2 \ref{##1}}
 \begin{rep@theorem}\itshape}
 {\end{rep@theorem}}}
\makeatother

\newreptheorem{theorem}{Theorem}

\newcommand{\Cert}{\mathrm{Cert}} 

\newcommand{\NP}{\mathsf{NP}} 
\renewcommand{\P}{\mathsf{P}} 
\newcommand{\coNP}{\mathsf{coNP}}

\begin{document}

\title{
Certification with an NP oracle
\vspace*{20pt} 
}

\author{ 
\hspace{15pt} Guy Blanc \vspace{6pt} \\
\hspace{15pt} {\small {\sl Stanford}} \and 
 Caleb Koch \vspace{6pt} \\ 
\hspace{-5pt} {\small {\sl Stanford}} \and 
\hspace{5pt} Jane Lange \vspace{6pt} \\ 
\hspace{5pt} {\small {\sl MIT}} \and 
\hspace{-5pt}  Carmen Strassle \vspace{6pt} \\
\hspace{-15pt}  {\small {\sl Stanford}} \and 
\hspace{-10pt} Li-Yang Tan \vspace{6pt}  \\
\hspace{-15pt} {\small {\sl Stanford}}
}

\date{\vspace{20pt}\small{\today}}

\maketitle

\begin{abstract}
In the {\sl certification problem}, the algorithm is given a function $f$ with certificate complexity~$k$ and an input $x^\star$, and the goal is to find a certificate of size $\le \poly(k)$ for $f$'s value at $x^\star$.  This problem is in $\mathsf{NP}^{\mathsf{NP}}$, and assuming $\mathsf{P} \ne \mathsf{NP}$, is not in $\mathsf{P}$.  Prior works, dating back to Valiant in 1984, have therefore sought to design efficient algorithms by imposing assumptions on $f$ such as monotonicity. 

Our first result is a  $\mathsf{BPP}^{\mathsf{NP}}$ algorithm for the general problem.  The key ingredient is a new notion of the {\sl balanced} influence of variables, a natural variant of influence that corrects for the bias of the function.  Balanced influences can be accurately estimated via uniform generation, and classic $\mathsf{BPP}^{\mathsf{NP}}$ algorithms are known for the latter task.

We then consider certification with stricter {\sl instance-wise} guarantees: for each $x^\star$, find a certificate whose size scales with that of the smallest certificate {\sl for $x^\star$}.   In sharp contrast with our first result, we show that this problem is $\mathsf{NP}^{\mathsf{NP}}$-hard even to approximate. We obtain an optimal inapproximability ratio, adding to a small handful of problems in the higher levels of the polynomial hierarchy for which optimal inapproximability is known.  Our proof involves the novel use of bit-fixing dispersers for gap amplification. 
\end{abstract}

\thispagestyle{empty}
\newpage 
\setcounter{page}{1}

\section{Introduction} 

For a function $f : \zo^n \to \zo$ and an input $x^\star \in \zo$, a {\sl certificate} for $f$'s value at $x^\star$ is a set $S\sse [n]$ of coordinates such that: 
\[ f(x^\star) = f(y) \quad \text{for all $y$ such that $y_S = x^\star_S$}. \] 
This is a set of coordinates that fully determines $f$'s value on $x^\star$, and coordinates outside this set are irrelevant in the sense that changing any of them in any way cannot change $f$'s value.   We write $\Cert(f,x^\star)$ to denote the size of a smallest certificate for $f$'s value at $x^\star$, and $\Cert(f)$ to denote  $\max_{x\in \zo^n}\{ \Cert(f,x)\}$, the certificate complexity of $f$.  Certificate complexity is among the most basic and well-studied measures of boolean function complexity~\cite{BdW02,Jukna12}.  

With the notion of certificates in mind, an algorithmic question suggests itself: can we design efficient algorithms for finding small certificates?   This leads us to the certification problem: 

\begin{quote}
{\bf Certification Problem:} Given the succinct description of a function $f$ and an input $x^\star\in \zo^n$, find a certificate of size $\le \poly(\Cert(f))$ for $f$'s value at $x^\star$. 
\end{quote}


\paragraph{The intractability of certification.} 
 Valiant was the first to consider the certification problem~\cite{Val84}. He observed that it is likely intractable in its full generality, since even the task of {\sl verifying} the output of a certification algorithm, i.e.~checking that a purported certificate is indeed a certificate, is {\sf coNP}-complete. 
 
Furthermore, the {\sf NP}-hardness of {\sc Sat} easily implies the following: 
\begin{fact}
\label{fact:intractable}
Assuming $\mathsf{P}\ne \mathsf{NP}$, there is no efficient algorithm for the certification problem, even if the algorithm only has to return a certificate of size $\Phi(\Cert(f))$ for any growth function $\Phi: \N \to \N$.  Similarly, we can rule out randomized algorithms under the assumption that $\mathsf{NP}\not\sse \mathsf{BPP}$. 
\end{fact}

For completeness, we include the proofs of Valiant's observation and~\Cref{fact:intractable} in~\Cref{sec:folklore}.

\paragraph{Prior certification algorithms.}  In light of these intractability results, prior certification algorithms have relied on assumptions on $f$.  Valiant gave a simple and efficient  algorithm for {\sl monotone} functions, which he used as a subroutine for PAC learning monotone DNF formulas (see also~\cite{Angluin88}). Recent works~\cite{BKLT22,GM22} give new certification algorithms for monotone functions that are furthermore highly query efficient, with the latter paper by Gupta and Manoj achieving the optimal query complexity.  In a separate line of work, Barcel\'o, Monet, P\'erez, and Subercaseaux~\cite{BMPS20} gave an efficient certification algorithm for {\sl halfspaces}. 

While the focus of these works and ours is theoretical in nature, there has also been a recent surge of interest in certification algorithms from an applied perspective.  Motivation here comes from the growing field of explainable machine learning, where one thinks of $f$ as a complicated model (e.g.~a neural net) and small certificates as succinct explanations for its decisions. In this literature, certificates are more commonly referred to as ``sufficient reasons"~\cite{SCD18} and ``anchors"~\cite{RSG18}; see~\cite{BMPS20,BLT-Explanations,BKLT22} for further discussions and references regarding this.

\subsection{Our results} 

Departing from the theme of prior works, we consider the certification problem in its full generality, without any assumptions on $f$, but allow algorithms that do not necessarily run in polynomial time.  There is a simple ${\sf NP}^{\mathsf{NP}}$ algorithm for general problem: first guess a certificate for $f$'s value at $x^\star$, and then use the {\sf NP} oracle to verify that it is indeed a certificate.  Therefore, the certification problem lies within (the function version of\footnote{Throughout this paper we conflate the decision and function versions of complexity classes, except in instances where there is a possibility of confusion.}) $\mathsf{NP}^{\mathsf{NP}}$, and assuming $\mathsf{P} \ne \mathsf{NP}$, lies outside of $\mathsf{P}$.  This leaves a rather wide gap between $\mathsf{P}$ and $\mathsf{NP}^{\mathsf{NP}}$.

At first glance, one may suspect that the true complexity of the certification problem is $\mathsf{NP}^{\mathsf{NP}}$. Afterall, the definition of certificate complexity inherently involves two quantifiers: {\sl there exists} a set $S$ such that $f(x^\star) = f(y)$ {\sl for all} $y$ that agrees with $x^\star$ on $S$.   Our first main result counters this intuition with a $\mathsf{BPP}^{\mathsf{NP}}$ algorithm:

\begin{theorem} 
\label{thm:upper bound intro} 
There is a randomized polynomial-time algorithm with $\NP$ oracle access that takes a polynomial-size circuit representation of $f:\zo^n\to\zo$ and a string $x^\star\in\zo^n$, and outputs w.h.p.~a certificate $S\sse [n]$ for $f$'s value on $x^\star$ satisfying $|S|= O(\Cert(f)^5)$. 
\end{theorem} 

Our algorithm does not need to be given the value of $\Cert(f)$ as an input.  {\color{black}{
As an example setting of parameters, for a function $f$ with certificate complexity $O(n^\eps)$ our algorithm always returns a certificate of size $O(n^{5\eps}) \ll n$. We note that the class of functions with certificate complexity $O(n^\eps)$ is very expressive and includes, among other functions, all decision trees of depth $O(n^{\eps})$. 
}}

Our algorithm uses its $\mathsf{NP}$ oracle in two ways: 
\begin{enumerate}
    \item {\sl Uniform generation.} Our algorithm uses as a key subroutine the ability to sample a uniform random satisfying assignment $\bx\sim f^{-1}(1)$ of $f$.  Early and influential results of theoretical computer science show that this can be done efficiently with an $\mathsf{NP}$ oracle~\cite{JVV86,BGP00}. 
    \item {\sl Verification.}  Our randomized algorithm returns a set that is a certificate for $f$'s value at $x^\star$ with high probability.  As mentioned, the task of verifying that a set is indeed a certificate is {\sf coNP}-complete, and so we use the $\mathsf{NP}$ oracle for this purpose.  
\end{enumerate}

Regarding the $\mathsf{BPP}^{\mathsf{NP}}$ algorithms for uniform generation, while they are not efficient in the traditional sense of worst-case complexity, there is an active line of research that seeks to make them as practical as possible: by replacing the $\mathsf{NP}$ oracle with state-of-the-art SAT solvers; optimizing the number of calls to these solvers; etc. See~\cite{Var,DM22} and the references therein.

\paragraph{Barriers to instance-wise guarantees.} In the spirit of beyond worst-case analysis, it is natural to ask if our algorithm can be strengthened so its guarantees hold {\sl instance-wise}:  can it always return a certificate of size $\poly(\Cert(f,x^\star))$ instead of $\poly(\Cert(f))$?  There are two known lower bounds  against such algorithms, both in the strictest setting where the algorithm has to return a certificate of size exactly $\Cert(f,x^\star)$.  Gupta and Manoj~\cite{GM22} showed that any such algorithm for monotone functions must make $n^{\Omega(\Cert(f,x^\star))}$ many queries to $f$, in the setting where the algorithm only gets black-box queries to $f$ rather than an explicit description.   Barcel\'o, Monet, P\'erez, and Subercaseaux~\cite{BMPS20} showed that the decision version of the problem, where the goal is to decide if $\Cert(f,x^\star)\le k$ for a given parameter $k$, is $\mathsf{NP}^{\mathsf{NP}}$-hard.  

Neither of these results rule out algorithms that return a certificate of size $\le \poly(\Cert(f,x^\star))$ or even $O(\Cert(f,x^\star))$.  Our second main result does so by showing that $\Cert(f,x^\star)$ is optimally  $\mathsf{NP}^{\mathsf{NP}}$-hard to approximate: 

\begin{theorem} 
\label{thm:hardness of approximation}
The following holds for every constant $\eps>0$. Given a circuit representation of $f:\zo^n\to\zo, k\in \N,$ and $x^\star\in \zo^n$, it is $\NP^{\NP}$-hard to distinguish between
\begin{itemize}
    \item[$\circ$] \textnormal{YES:} there exists a certificate of size $\le k$ for $f$'s value on $x^\star$;
    \item[$\circ$] \textnormal{NO:} Every certificate for $f$'s value on $x^\star$ has size $>k\cdot n^{1-\eps}$.
\end{itemize}
\end{theorem}

There is by now a sizeable number of problems that are known to be hard for higher levels of the polynomial hierarchy~\cite{SU02a,SU02b}.  However, relatively few of them are known to be hard to {\sl approximate}, and yet fewer for which optimal inapproximability ratios have been established.  The papers~\cite{Uma99,Uma99b,Uma00,Uma01,TSUZ01,MU02} cover many of the existing results; see also the survey~\cite{Uma06}. \Cref{thm:hardness of approximation} thus adds an additional entry into this catalogue of natural problems known to be inapproximable for higher levels of the polynomial hierarchy.

\section{Proof overviews}
\label{sec:techniques}

\subsection{Overview of the proof of~\Cref{thm:upper bound intro}}

\paragraph{Minimality no longer suffices.}  As mentioned, many prior certification algorithms have focused on the class of {\sl monotone} functions~\cite{Val84,Angluin88,BKLT22,GM22}.  An especially nice feature of this setting is that it suffices to find {\sl minimal} certificates: a certificate $S$ such that no strict subset $S'\subsetneq S$ is a certificate:

\begin{fact}[Minimality suffices for monotone functions] 
\label{fact:monotone-minimal} 
If $f$ is monotone, for {\sl every} $x^\star$, {\sl every} minimal certificate $S$ for $f$'s value at $x^\star$ has size $|S| \le \Cert(f)$.
\end{fact}

This simple fact is crucially used in all prior certification algorithms for monotone functions.  While $\Cert(f)$ is a global property of $f$, the minimality of a specific certificate can be recognized locally: if at any point we have a candidate certificate $S$ such that dropping any of its coordinates results in it no longer being a certificate,~\Cref{fact:monotone-minimal} tells us that $|S| \le \Cert(f)$ and we are done.   

Unfortunately, \Cref{fact:monotone-minimal} is false for general functions:  the size of a minimal certificate could be exponentially larger than its certificate complexity:  

\begin{fact}[{\cite[Exercise 1.7]{Jukna12}}]
\label{fact:minimal-certificate-vs-c(f)-gap}
For each $n\in\N$, there is a function $f:\zo^n\to\zo$ and $x^\star\in\zo^n$ such that $x^\star$ has a minimal certificate $S\sse [n]$ of size $|S| \ge \Omega(n)$ even though $\Cert(f) \le O(\log n)$. 
\end{fact}

The implication of~\Cref{fact:minimal-certificate-vs-c(f)-gap} for the certification problem is that algorithms for general functions can get stuck at (very bad) local minima, which necessitates a more global understanding of the structure of functions with low certificate complexity.  We include the proofs of~\Cref{fact:monotone-minimal} and~\Cref{fact:minimal-certificate-vs-c(f)-gap} in~\Cref{sec:folklore}.

\paragraph{Our algorithm and its main subroutine.}  We now describe our algorithm.  Its main subroutine is a $\mathsf{BPP}^{\mathsf{NP}}$ algorithm for finding a restriction to a small number of variables under which $f$ becomes constant: 

\begin{lemma}
\label{lem:certify arbitrary input}
Given a succinct representation of $f : \zo^n\to \zo$ and an $\mathsf{NP}$ oracle, there is a randomized polynomial-time algorithm that w.h.p.~finds a restriction $\pi$ to $O(\Cert(f)^4)$ variables such that $f_\pi$ is a constant function. 
\end{lemma}

Such an algorithm can be viewed as one that finds a certificate for {\sl some} input of $f$ (any of the inputs that are consistent with $\pi$), but not necessarily the specific input $x^\star$ that we are interested in certifying.  We then convert such an algorithm into an actual certification algorithm by calling it $2\,\Cert(f)$ times, resulting in a certificate for $x^\star$ of size $O(\Cert(f)^5)$.  This conversion algorithm is a fairly straightforward consequence of a basic property of certificates, that every $1$-certificate and every $0$-certificate must share at least one variable.  


\paragraph{Balanced influences.} To describe our algorithm for~\Cref{lem:certify arbitrary input}, we need a new notion of the {\sl balanced influence} of variables. Recall that the {\sl influence} of a variable $i$ on $f$ is the quantity $\Inf_i(f) \coloneqq  \Prx_{\bx\sim\zo^n}[f(\bx)\ne f(\bx^{\oplus i})]$, where $\bx$ is uniform random and $\bx^{\oplus i}$ denotes $\bx$ with its $i$-th bit flipped. It is easy to see that $\Inf_i(f) \le \Var(f)$, and so the more unbalanced $f$ is, the smaller its influences are.  Balanced influence corrects for this by measuring influence with respect to a distribution $\mathcal{D}_{\bal}^{(f)}$ that places equal weight on $f^{-1}(1)$ and $f^{-1}(0)$:

\begin{center} 
To sample $\bx \sim \mathcal{D}_{\bal}^{(f)}$: first sample $\bb \sim \zo$ uniformly and then $\bx \sim f^{-1}(\bb)$ uniformly.
\end{center} 

\begin{definition}[Balanced influences]
    \label{def:balanced-influences}
     The \emph{balanced influence} of a variable $i \in [n]$ on $f : \zo^n \to \zo$ is the quantity: 
    \begin{equation*}
        \balinf_i(f) \coloneqq \Prx_{\bx \sim \mathcal{D}_{\bal}^{(f)}}[f(\bx) \neq f(\bx^{\oplus i})].
    \end{equation*}
\end{definition}

Though simple, this variant of influence does not appear to have been explicitly considered before.  This definition syncs up nicely with the notion of uniform generation: an algorithm for uniform generation can be used to efficiently estimate balanced influences to high accuracy.  Given the classic $\mathsf{BPP}^{\mathsf{NP}}$ algorithms for uniform generation~\cite{JVV86,BGP00}, we therefore get: 

\begin{lemma} 
\label{lem:est-bal-inf}
There is a randomized algorithm that, given a succinct representation of $f : \zo^n \to \zo$, an $\mathsf{NP}$ oracle, $i\in [n]$ and $\eps > 0$, runs in $\poly(n,1/\eps)$ time and w.h.p.~outputs an estimate $\boldsymbol{\eta}_i = \balinf_i(f) \pm \eps$.  
\end{lemma}

We can now state our algorithm for~\Cref{lem:certify arbitrary input}. It is simple and proceeds by iteratively and randomly restricting the variables of $f$ with the largest balanced influence:

\begin{enumerate} 
\item Estimate $\balinf_i(f)$ for each $i\in [n]$ to within $\pm \frac1{4n}$ and let $i^\star$ be the index with the largest estimate. \item Randomly restrict $x_{i^\star} \leftarrow \bb$ where $\bb\sim \zo$ is uniform random. 
\item Recurse on $f_{x_{i^\star} = \bb}$. 
\end{enumerate} 

We prove that $f$ becomes constant w.h.p.~within ${O}(\Cert(f)^5)$ many iterations of this algorithm.   It is crucial that we work with {\sl balanced} influence here: this same algorithm fails (i.e.~it requires $\Omega(n)$ many iterations before $f$ becomes constant) if it is instead run on estimates of the usual notion of influence.\footnote{Indeed, since the usual notion of influence can be estimated to high accuracy without an $\mathsf{NP}$ oracle, we have by~\Cref{fact:intractable} that {\sl no} algorithm that is based only on estimates of usual influences can succeed unless $\mathsf{NP}\sse \mathsf{BPP}$.}

\subsection{Overview of the proof of~\Cref{thm:hardness of approximation}}

\paragraph{Monotone Minimum Weight Word and its inapproximability.} We prove~\Cref{thm:hardness of approximation} by reducing from the {\sc Monotone Minimum Weight Word} problem.  To state this problem we first define {\sl co-nondeterministic} circuits:

\begin{definition}
[Co-nondeterministic circuit] \label{def:nondet-circuit} 
A {\sl co-nondeterministic circuit} $C : \zo^n \times \zo^n \to \zo$ accepts an input $x\in \zo^n$ iff for all $y\in \zo^n$ $C(x,y)=1$. 
\end{definition}

We say that a co-nondeterministic circuit $C$ {\sl accepts a non-empty monotone set} if it accepts at least one $x$, and for every $x$ that
it accepts, it also accepts every $x'$ such that $x'\succeq x$.  (Note that $C : \zo^n \times\zo^n\to \zo$ itself need not be a monotone function.)   

\begin{definition} [{\sc Monotone Minimum Weight Word}] \label{def:MMWW} 
The {\sc Monotone Minimum Weight Word} $(\mmww)$ problem is the following: given a co-nondeterministic circuit $C : \zo^n \times \zo^n \to \zo$ that accepts a non-empty monotone set and an integer $k \in \N$, does $C$ accept an input $x\in \zo^n$ with at most $k$ ones?
\end{definition}

Umans~\cite{Uma99} showed that the $\mmww$ problem is $\mathsf{NP}^{\mathsf{NP}}$-hard to approximate within $n^{\frac1{5}-\eps}$. This inapproximability ratio was subsequently improved to the optimal $n^{1-\eps}$ by Ta-Shma, Umans, and Zuckerman~\cite{TSUZ01}: 

\begin{theorem}
[{$\gmmww$} is   $\mathsf{NP}^{\mathsf{NP}}$-hard] \label{thm:umanslb} 
The following holds for every constant $\eps > 0$. Given a co-nondeterministic circuit $C : \zo^n \times \zo^n \to \zo$ that accepts a non-empty monotone set and an integer $k\in \N$, it is $\mathsf{NP}^{\mathsf{NP}}$-hard to distinguish between: 
\begin{itemize} 
\item[$\circ$] \textnormal{YES:} $C$ accepts an $x$ with $\leq k$ ones; 
\item[$\circ$] \textnormal{NO:} Every $x$ that $C$ accepts has $>  k\cdot n^{1-\eps}$ ones. 
\end{itemize} 
\end{theorem}

\paragraph{First attempt at a reduction.}  To describe the intuition behind our reduction, we first consider what happens when we take an instance of $\gmmww$, which is specified by a co-nondeterministic circuit $C : \zo^n \times \zo^n \to \zo^n$ and an integer $k$, and map it to the certification problem with $f(x,y) = C(x,y)$ and the input to be certified being $(x^\star,y^\star) = (1^n,1^n)$.  We consider the two possible cases: 

\begin{itemize} 
\item[$\circ$] If $(C,k)$ is a YES instance of $\gmmww$, it accepts an $x \in \zo^n$ with $\le k$ ones. Let $S\sse [n]$ be the size-$k$ set of $1$-coordinates of this input $x$.  It is straightforward to verify that $S$ is a certificate for $f$'s value on $(x^\star,y^\star)$, and so $\Cert(f,(x^\star,y^\star))\le k$. 
\item[$\circ$] If $(C,k)$ is a NO instance of $\gmmww$, we would like it to be the case that $\Cert(f,(x^\star,y^\star)) > k\cdot n^{1-\eps}$.  Let $S \sse [n] \times [n]$ be a smallest certificate for $f$'s value on $(x^\star,y^\star) $.  If $S$ only contains coordinates of $x^\star$, it is again straightforward to verify that $C$ accepts the input $x\in\zo^n$ that is the indicator vector of $S$, and so $|S| > k\cdot n^{1-\eps}$ as desired. However, $S$ may contain one of more coordinates of $y^\star$, and in that case we do not have a lower bound on its size. 
\end{itemize}

\paragraph{The actual reduction.}  Intuitively, we would like to fix this issue by modifying $C$ so that it becomes disproportionately more ``expensive" to include $y^\star$-coordinates in the certificate compared to $x^\star$-coordinates, in the sense that including even a single $y^\star$-coordinate contributes $> k\cdot n^{1-\eps}$ to the size of the certificate, compared to just one in the case of single $x^\star$-coordinate.  We accomplish this with this use of {\sl bit-fixing dispersers}, a well-studied construct in the pseudorandomness literature. 

Given an instance $C : \zo^n \times \zo^n\to \zo$ of $\gmmww$, we map it to an instance of the certification problem where the function $f : \zo^n\times \zo^m \to \zo^n$ is:  
\[ f(x,z) \coloneqq C(x,D(z)), \] 
and $D : \zo^m \to \zo^n$ satisfies the following properties: 
\begin{enumerate} 
\item {\bf $D$ is explicit.}  There is an efficient algorithm that, given $m$ and $n$, produces the circuit description of $D : \zo^m\to\zo^n$ in $\poly(m,n)$ time. This is so that our reduction is efficient.   
\item {\bf $D$ retains full image under restrictions.}  For any set $S \sse [m]$ of size $\le k\cdot n^{1-\eps}$ and any assignment $u\in \zo^{|S|}$, the image of $D_{S\leftarrow u}$ is all of $\zo^n$: for any $y\in \zo^n$ there is an $z\in \zo^{m-|S|}$ such that $D_{S\leftarrow u}(z) = y$.  This ensures that one must fix $> k\cdot n^{1-\eps}$ many $z$-variables of $f$ in order to fix even a single $y$-variable of $C$. 
\item {\bf Small $m$.}  For the second property to hold $m$ certainly has to be at least $n$.  Since $f$ is a function over $m+n$ variables, would like $m$ to be as close to $n$ as possible in order to preserve the $n^{1-\eps}$ ratio of $\gmmww$. 
\end{enumerate} 

A simple construction of a function satisfying the first two properties is the blockwise parity function: 
\[ \mathrm{BlockwisePar} : (\zo^\ell)^n \to \zo, \] 
\[ \mathrm{BlockwisePar}(z^{(1)},\ldots,z^{(n)})^n = (\oplus_{j=1}^\ell z^{(1)}_j,\cdots,\oplus_{j=1}^\ell z^{(n)}_j)\] 
where $\ell = k\cdot n^{1-\eps}+1$. In this case $m = \ell\cdot n = \Theta(n^{2-\eps}),$ and so the $n^{1-\eps}$ ratio of $\gmmww$ translates into a gap of $n^{\frac1{2}-\eps}$ for the certification problem. 

Functions satisfying the first two properties are known as zero-error bit-fixing dispersers in the pseudorandomness literature.  The current best  construction, due to Gabizon and Shaltiel~\cite{GS12}, gives $m = O(n)$ for our setting of parameters, and so using it in place of BlockwisePar enables us to achieve the optimal inapproximability ratio of $n^{1-\eps}$, thereby yielding~\Cref{thm:hardness of approximation}. 

\section{Discussion and future work}

Since the certification problem is intractable unless $\mathsf{P} =  \mathsf{NP}$, to further understand it we must either rely on assumptions about $f$ or allow algorithms that are not necessarily efficient in the traditional sense of running in polynomial time.  Complementing previous works that take the former route, in this work we consider the latter option and study the complexity of certification {\sl relative to an $\mathsf{NP}$ oracle}, giving new algorithmic ($\mathsf{BPP}^{\mathsf{NP}}$) and hardness ($\mathsf{NP}^{\mathsf{NP}}$) results.  Our motivation is twofold.  First, given how natural the problem is, we believe that it is of independent interest to understand its inherent complexity.  Second, given the empirical success of SAT solvers and their increasing adoption in a variety of real-world algorithmic tasks, the distinction between between problems in $\mathsf{BPP}^{\mathsf{NP}}$ (``exists an efficient algorithm assuming all calls to the SAT solver run quickly") versus those that are hard for  $\mathsf{NP}^{\mathsf{NP}}$ (``intractable even if all calls to the SAT solver run in unit time") is especially relevant in this context.   

{\color{black}{
We list a couple of concrete open problems suggested by our work. A natural one is to improve on the bound of $O(\Cert(f)^5)$ given by~\Cref{thm:upper bound intro}:  

\begin{openproblem} 
Is there a $\mathsf{BPP}^{\mathsf{NP}}$ certification algorithm that returns certificates of length $O(\Cert(f))$? 
\end{openproblem} 
}}

Next,~\Cref{thm:hardness of approximation} shows that given $f$ and $x^\star$, it is optimally $\mathsf{NP}^{\mathsf{NP}}$-hard to approximate the size of the smallest certificate for $f$'s value at $x^\star$, i.e.~to approximate the quantity $\Cert(f,x^\star)$.  Since there is an $\mathsf{NP}^{\mathsf{NP}}$ algorithm that computes it exactly, this settles the complexity of the problem.  It is equally natural to consider the problem of computing/approximating $\Cert(f)$.  There is a simple $\Pi_3$ algorithm that computes it exactly. However, little is known in terms of lower bounds: 

{\color{black}{ 
\begin{openproblem}
Is certificate complexity $\Pi_3$-hard to compute and what is the complexity of approximating this quantity?  
\end{openproblem}
}}

Concluding on a speculative note, in the spirit of {\sl interactive proofs}, it would be interesting to extend our $\mathsf{BPP}^{\mathsf{NP}}$ algorithm to the setting where the algorithm interacts with a powerful but {\sl untrusted} oracle. More broadly, there should be much to be gained by bringing techniques from interactive proofs to bear on problems, such as certification, that are motivated by explainable machine learning. Such an ``interactive theory of explanations" was listed as a specific direction in the 2020 TCS Visioning Report~\cite{CNUW21} and aligns well with ongoing efforts in machine learning~\cite{WB19}, but to our knowledge has thus far not been explored much.



\section{Basic results regarding the certification problem}
\label{sec:folklore}

\subsection{The necessity of an $\NP$ oracle}

\paragraph{Verifying a certificate is $\coNP$-complete.}  Given a candidate certificate for an input, verifying the certificate is computationally hard. Indeed, we observe that such a problem is $\coNP$-complete. This lemma shows that we at least need access to a \textit{$\coNP$} oracle or equivalently an $\NP$ oracle.


\begin{definition}[{\sc VerifyCert}]
    Given a polynomial-size circuit for $f: \zo^n \to \zo$, input $x \in \zo^n$, and candidate certificate $S \subseteq [n]$, decide whether $S$ is indeed a certificate for $f$'s classification of $x$.
\end{definition}

\begin{lemma}[Observed in~\cite{Val84}]
    \label{lem:verifycert-conp-complete}
    {\sc VerifyCert} is $\coNP$-complete.
\end{lemma}

To prove this lemma, we give a reduction from the canonical $\coNP$-complete problem \textsc{Tautology}. 

\begin{definition}[{\sc Tautology} \cite{Cook71}]
    Given a Boolean formula $\varphi:\zo^n\to\zo$, decide whether $\varphi$ is a tautology: $\varphi(x)=1$ for all $x\in\zo^n$.
\end{definition}

\begin{proof}[Proof of \Cref{lem:verifycert-conp-complete}]
First, we show that $\textsc{VerifyCert}\in\coNP$. Specifically, we observe that $\overline{\textsc{VerifyCert}}\in\NP$. Given $f:\zo^n\to\zo$, an input $x\in\zo^n$, and a candidate $S\sse [n]$, one can nondeterministically guess a string $y$ such that $x_S=y_S$ and $f(x)\neq f(y)$. This string $y$ exists if and only if $S$ is \textit{not} a certificate for $x$.

Next, we show $\textsc{VerifyCert}$ is $\coNP$-hard via a reduction from \textsc{Tautology}. Let $\varphi:\zo^n\to\zo$ be a Boolean formula. We show how to determine if $\varphi$ is a tautology using $\textsc{VerifyCert}$. Let $x\in\zo^n$ be arbitrary. If $\varphi(x)=0$ output ``not a tautology''. Otherwise, call $\textsc{VerifyCert}$ on the instance $f=\varphi, S=\varnothing$ and $x$. If $\varnothing$ is a certificate for $x$ then $\varphi$ is a tautology: all $y\in\zo^n$ satisfy $\varphi(y)=\varphi(x)=1$. Otherwise, there is some input $y$ for which $\varphi(y)=0$. 
\end{proof}


\paragraph{Solving the certificate finding problem efficiently without an $\NP$ oracle would prove $\P = \NP$ (\Cref{fact:intractable}).}  Next, we show that any efficient algorithm for certification can be used to solve {{\sc GappedCert}}$(0,\Psi)$: the promise problem of distinguishing whether $\Cert(f)$ is $0$ or every input requires size-$\Psi(n)$ certificates. Perhaps surprisingly, {\sc GappedCert}$(0,\Omega(n))$ turns out to be $\NP$-hard. As a result, we should not expect certification to be efficiently solvable without access to an $\NP$ oracle.

\begin{definition}[$\textsc{GappedCert}(0,\Psi)$]
\label{def:gapped-cert}
Given a polynomial-size circuit for $f:\zo^n\to\zo$ and a growth function $\Psi:\N\to\N$, the $\textsc{GappedCert}(0,\Psi)$ problem is to distinguish between
\begin{itemize}
    \item[$\circ$] \textnormal{YES:} $\Cert(f,x)\ge \Psi(n)$ for all $x\in\zo^n$;
    \item[$\circ$] \textnormal{NO:} $\Cert(f)=0$.
\end{itemize}
\end{definition}

As a promise decision problem, any algorithm for $\textsc{GappedCert}(0,\Psi)$ is allowed to answer arbitrarily on input instances which do not fall into the two cases. Note that SAT reduces to $\textsc{GappedCert}(0,1)$ since $\textsc{GappedCert}(0,1)$ is equivalent to determining whether $f$ is constant. To obtain hardness against the certification problem, we require hardness for \textsc{GappedCert} where the growth function is nonconstant. 

\begin{lemma}
    \label{lem:sqrt-gapped-cert-is-np-complete}
    $\textsc{GappedCert}(0,\sqrt{n})$ is $\NP$-hard.
\end{lemma}

\begin{proof}
Let $\varphi:\zo^n\to\zo$ be a Boolean formula. Recall the definition of $\mathrm{BlockwisePar}:(\zo^\ell)^n\to\zo^n$
$$
\mathrm{BlockwisePar}(z^{(1)},\ldots,z^{(n)})=(\oplus_{j=1}^\ell z_j^{(1)},\ldots,\oplus_{j=1}^\ell z_j^{(n)}).
$$
We let $f=\varphi\circ\mathrm{BlockwisePar}$. For all $z\in (\zo^\ell)^n$, we have 
$$
\Cert(f,z)\ge \Cert(\varphi,\mathrm{BlockwisePar}(z))\cdot\ell
$$
since $\ell$ coordinates of $f$ need to be fixed in order to fix one coordinate of $\varphi$. In particular, if $\varphi$ is nonconstant, then $\Cert(\varphi,\mathrm{BlockwisePar}(z))\ge 1$ and so $\Cert(f,z)\ge \ell$ for all $z\in (\zo^\ell)^n$. 

Therefore, to determine satisfiability of $\varphi$ it is sufficient to solve $\textsc{GappedCert}(0,\ell)$ for the function $f:(\zo^\ell)^n\to\zo$. Choosing $\ell=n$ yields the hardness in the lemma statement.
\end{proof}

\Cref{lem:sqrt-gapped-cert-is-np-complete} is sufficient to establish \Cref{fact:intractable}. However, we also note that a slightly more technical proof yields near-optimal gapped hardness.

\begin{lemma}[Near-optimal gapped hardness]
    \label{lem:n-gapped-cert-is-np-complete}
    $\textsc{GappedCert}(0,\Omega(n))$ is $\NP$-hard.
\end{lemma}

We provide a proof of \Cref{lem:n-gapped-cert-is-np-complete} in \Cref{appendix:gappedcert-hardness}.

\begin{lemma}
\label{lem:certification-yields-gapped-cert-algo}
    If there is a polynomial-time algorithm for the certification problem which returns certificates of size $\Phi(\Cert(f))$ for some growth function $\Phi:\N\to\N$, then there is a polynomial-time algorithm for $\textsc{GappedCert}(0,\Psi)$ for all growth functions $\Psi:\N\to\N$. 
\end{lemma}

\begin{proof}
Suppose such an algorithm exists for the certification problem. Let $\Psi:\N\to\N$ be an arbitrary growth function and let $f:\zo^n\to\zo$ be an instance of $\textsc{GappedCert}(0,\Psi)$. Let $S\sse [n]$ be a certificate obtained by running the algorithm for the certification problem on $f$ for any input $x^\star\in\zo^n$. Output ``NO'' if $|S|<\Psi(n)$ and ``YES'' if $|S|\ge \Psi(n)$. 

The correctness of our output follows from the observation that $\Cert(f,x)\le |S|\le \Phi(\Cert(f))$ and so
\begin{align*}
    |S|&\ge \Psi(n)\quad \Rightarrow \quad\Psi(n)\le \Phi(\Cert(f))\\
    |S|&<\Psi(n)\quad \Rightarrow \quad\Cert(f,x)<\Psi(n).
\end{align*}
Assuming $n$ is large enough, the first case implies $\Cert(f)>0$ which allows us to rule out being in the \textnormal{NO} case of $\textsc{GappedCert}(0,\Psi)$. The second case rules out being in the \textnormal{YES} case of $\textsc{GappedCert}(0,\Psi)$.
\end{proof}

\Cref{fact:intractable} follows as an immediate consequence of \Cref{lem:sqrt-gapped-cert-is-np-complete,lem:certification-yields-gapped-cert-algo}: any polynomial-time algorithm for the certification problem yields a polynomial-time algorithm for $\textsc{GappedCert}(0,\sqrt{n})$ which implies $\NP=\P$. Likewise, any randomized polynomial-time algorithm for the certification problem yields a randomized polynomial-time algorithm for $\textsc{GappedCert}(0,\sqrt{n})$ which implies $\NP\sse \mathsf{BPP}$.

\subsection{Certificate complexity, minimal certificates, and monotonicity} 

In this section, we prove \Cref{fact:monotone-minimal} and \Cref{fact:minimal-certificate-vs-c(f)-gap}. We start with a formal definition of \textit{minimal certificates}.

\begin{definition}[Minimal certificates\footnote{Minimal certificates are also called minterms/maxterms as in e.g. \cite{Jukna12}. A minterm is a minimal $0$-certificate and a maxterm is a minimal $1$-certificate.}]
A certificate $S\sse [n]$ for $f:\zo^n\to\zo$ on $x\in\zo^n$ is a \textit{minimal} certificate if no $S'\subsetneq S$ is a certificate for $f$ on $x$. 
\end{definition}

\paragraph{Proof of \Cref{fact:monotone-minimal}}{
Suppose $S\sse [n]$ is a minimal $1$-certificate for $f$'s value on a string $x^\star\in\zo^n$ (if $S$ is a $0$-certificate the argument is symmetric). First, we observe that $S$ can only contain $1$-coordinates of $x^\star$. Indeed, if $S$ contained a $0$-coordinate, the certificate $S'$ obtained from $S$ by removing that $0$-coordinate would still certify $x^\star$ by monotonicity and would therefore contradict the minimality of $S$. Consider the string $x'\in\zo^n$ formed by setting all the coordinates in $S$ to $1$ and all the coordinates outside $S$ to $0$. Then, $|S|=\Cert(f,x')\le \Cert(f)$. Specifically, any minimal certificate for $f$'s value on $x'$ must be a subset of its $1$-coordinates and hence a subset of $S$. The minimality of $S$ implies no $S'\subsetneq S$ can certify $f$'s value on $x'$.\hfill\qed
}

\paragraph{Proof of \Cref{fact:minimal-certificate-vs-c(f)-gap}.}{
This fact is witnessed by the \textit{addressing function}, $\textsc{Address}_r:[r]\times\zo^{r}\to\zo$ defined as $\textsc{Address}_r(x,y)=y_x$. Viewing $\textsc{Address}_r$ as a Boolean function on $\log r+r$ bits, we have $\Cert(\textsc{Address}_r)=\log r+1$ but any fixing of the $r$ bits of $y$ constitutes a minimal certificate.\hfill\qed
}


\section{A structural result for functions with low certificate complexity}
\label{sec:structural}

{\bf Notation and useful definitions.} 
All distributions over $\zo^n$ are uniform unless otherwise specified. We use {\bf boldface} to denote random variables (e.g. $\bx\sim\zo^n$) and we write ``w.h.p." to mean with probability $\ge 1-1/n^{\omega(1)}$.

\begin{definition}[Variance] 
    For any function $f: \zo^n \to \zo$, the variance of $f$ is defined as
    \begin{equation*}
        \Var(f) \coloneqq \underset{\bx\sim\zo^n}{\mathrm{Var}}
        [f(\bx)] = \Prx_{\bx \sim \zo^n}[f(\bx) = 0] 
        \cdot \Prx_{\bx \sim \zo^n}[f(\bx) = 1].
    \end{equation*}
\end{definition}

\begin{definition}[Influence]
    For any functions $f: \zo^n \to \zo$, we define the \emph{influence} of the coordinate $i \in [n]$ to be
    \begin{equation*}
        \Inf_i(f) \coloneqq \Prx_{\bx \sim \zo^n}[f(\bx) \neq f(\bx^{\oplus i})]
    \end{equation*}
    where $x^{\oplus i} \in \zo^n$ is the unique point that differs from $x$ in only the $i^{\text{th}}$ coordinate. We also define the \emph{total influence} as
    \begin{equation*}
        \Inf(f) \coloneqq \sum_{i \in [n]} \Inf_i(f),
    \end{equation*}
    and the \emph{maximum influence} as
    \begin{equation*}
        \maxinf(f) \coloneqq \max_{i \in [n]} \left\{  \Inf_i(f)\right\}.
    \end{equation*}
\end{definition}

%
%
%
%
%
%
\paragraph{Main structural result.} In this section we prove the following structural result:  iteratively and randomly restricting a function $f$ by its most influential variable causes it to become constant w.h.p.~after \violet{$O(\Cert(f)^4)$} iterations. To do so, we analyze the following family of decision trees. 
 
\begin{definition}
\label{def:inf-max-tree}
For a function $f : \zo^n\to\zo$, the {\sl influence-maximizing tree of depth $d$ for $f$}, 
denoted $T_f(d)$, is the complete tree of depth $d$
such that for each internal node $v$, 
the variable $i(v)$ queried is the one with the largest influence in the 
subfunction $f_v$: 
\[ \Inf_{i(v)}(f_v) \ge \Inf_j(f_v) \quad \text{for all $j\in [n]$}. \] 
Ties are broken arbitrarily, 
and for convenience, variables are still queried even if the function is already constant before depth $d$. 
\end{definition}

\violet{See \Cref{fig:path-in-idt} for an illustration of the influence-maximizing tree. We interpret the leaves and internal nodes of $T_f(d)$ as functions obtained from $f$ by restricting each variable corresponding to an edge in that node's path.} Note that a random leaf $\bell \sim T_f(d)$ corresponds a random restriction where we iteratively and randomly restrict the most influential variable of $f$ for $d$ iterations.

\begin{lemma}
\label{lem:most-leaves-constant}
Let $f$ be a function with certificate complexity $\Cert(f) \le k$. 
Then, for a leaf $\bell$ drawn uniformly at random from  $T_f(d)$ for $d = O(k^4 + k^3 \log (1/\eps))$, we have 
\[ \Prx_{\bell}\big[ \text{$f_{\bell}$ is constant}\,\big] \ge 1-\eps.  \] 
\end{lemma}

\begin{figure}[h!]
    \centering
    \begin{tikzpicture}[tips=proper]
        \node[isosceles triangle,
            draw,
            isosceles triangle apex angle=60,
            rotate=90,
            minimum size=6cm] (T1) at (0,0){};
        \draw[black,dotted] (T1.east) .. controls ([xshift=-0.1cm]T1.357) .. ([yshift=-0.5cm]T1.east) node[] (N1) {{}};
        \draw[black,dotted] ([yshift=-0.5cm]T1.east) .. controls ([xshift=0.4cm]T1.40) and (T1.350) .. (T1.center) node[fill=white,pos=0.8] (N1) {$\pi$};
        
        \draw[{Stealth[scale=1]}-{Stealth[scale=1]}] ([xshift=-5cm]T1.east) to node[midway,fill=white!30,scale=1] {depth ${O}(k^4)$} ([xshift=-5cm]T1.west);
        
        \draw[color=black] (T1.center) node [below,fill=white] {$x_i$};
        \node[draw,circle,fill=black,inner sep=1pt] (x) at (T1.center) {};
        

        
        \draw[{Stealth[scale=1]}-] (x) to ([xshift=1.5cm]T1.right side);
        
        \draw[] ([xshift=1.5cm]T1.right side) node [right,black,fill=white, text width=6cm] {For every path $\pi$:\\ $\Inf_i(f_\pi)=\maxinf(f_\pi)$};
        
        \draw [black,decorate,decoration={brace,mirror,raise=1pt,amplitude=3pt}] (T1.left corner) -- ([xshift=2.4cm]T1.left corner) node [black,pos=0.5,yshift=-0.5cm] {\small $\le 2^{-k}$-fraction};
        
        \draw [] ([xshift=2.3cm]T1.left corner) -- (T1.right corner) node [black,pos=0.5,below] {\small $~0~1~1~0~1~0~0~1~0~1~1~0~0~1~1$};
    \end{tikzpicture}
    \caption{Illustration of a decision tree constructed in a top-down greedy fashion by iteratively querying variables that maximizes influence. All but a $2^{-k}$-fraction of paths in such a tree have depth ${O}(k^4)$ where $k=\Cert(f)$.}
    \label{fig:path-in-idt}
\end{figure}
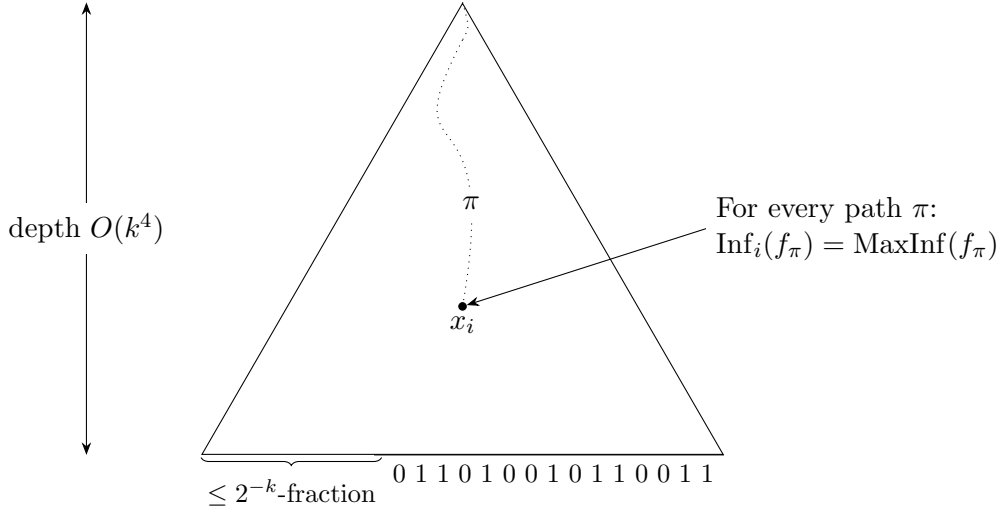


\subsection{Useful facts for the proof  of~\Cref{lem:most-leaves-constant}}
The proof of \Cref{lem:most-leaves-constant} is by a potential function argument,
where the potential function is the average total influence 
of leaf functions $f_{\bell}$ in $T_f(d)$: 
\[\phi(d) := \E_{{\bell} \sim T_f(d)}[\Inf(f_{\bell})].\]
In the remainder of this subsection, we will prove an upper bound 
on the initial value $\phi(0)$ and a lower bound 
on the drop in $\phi$ as we increment $d$. 
To do so, we use some basic facts about certificates and query complexity, and 
a well-known inequality of \cite{OSSS05}.

\newcommand{\Depth}{\mathrm{Depth}}

\begin{definition}[Deterministic query complexity]
The depth, or deterministic query complexity of $f$ is the depth of the minimum-depth decision tree computing $f$: 
\[\Depth(f) := \min_{T \text{ computing } f} \max_{\ell \in T} [\mathrm{depth}(\ell)].\]
\end{definition}

\violet{Certificate complexity and query complexity are well known to be within a polynomial factor of one another. See e.g. \cite{BdW02}.}
\begin{fact}
\violet{For any $f: \zo^n \to \zo$,} $\Cert(f) \le \Depth(f) \le \Cert(f)^2$.
\end{fact}

\violet{We'll need to lower and upper bound $\Inf(f)$ as a function of $\Var(f)$.}
\begin{proposition}
\label{lem:basics}
\violet{For any $f: \zo^n \to \zo$,} $\violet{4}\Var(f) \le \Inf(f) \le 4\Var(f) \cdot \Cert(f)$.
\end{proposition}
\begin{proof}
The first inequality is the well-known Poincaré inequality for the 
Boolean cube \cite{ODBook}. For the second, let $\Sens(f,x)$ be the number sensitive coordinates of $x$, i.e. the number distinct $i \in [n]$ for which $f(x) \neq f(x^{\oplus i})$. We rewrite the definition for the total influence of $f$,
\begin{align*}
    \Inf(f) &= \Ex_{\bx \sim \zo^n}[\Sens(f,\bx)] \\
    &=\Ex_{\bx \sim \zo^n}[\Sens(f,\bx) \cdot \Ind[f(\bx) = 0] + \Sens(f,\bx) \cdot \Ind[f(\bx) = 1]]. 
\end{align*}
Every sensitive edge (i.e. $(x, x^{\oplus i})$ where $f(x) \neq f(x^{\oplus i})$) must contain one endpoint classified as $0$ and one endpoint classified as $1$. Therefore, the two terms in the above sum are equal. Furthermore, any certificate for $x$ must include all $\Sens(f,x)$ sensitive coordinates, so $\mathrm{Sens}(f,x) \leq \Cert(f,x) \leq \Cert(f)$. We can therefore bound, 
\begin{align*}
    \Inf(f) &= 2 \cdot \min\big(\Ex[\Sens(f, \bx) \cdot \Ind[f(\bx) = 0]], \Ex[\Sens(f, \bx) \cdot \Ind[f(\bx) = 1]] \big) \\
     &\leq 2 \cdot \Cert(f) \cdot \min(\Prx[f(\bx) = 0], \Prx[f(\bx)= 1]) \\
     &\leq 4\cdot \Cert(f) \cdot \Prx[f(\bx) = 0] \cdot \Prx[f(\bx)= 1] = 4 \cdot \Cert(f) \cdot \Var(f). \qedhere
\end{align*} 

\end{proof}

Using the bound $\Var(f) \leq \frac{1}{4}$ for any $f$ with range $\zo$, we obtain the following corollary.
\begin{corollary}
\label{cor:total-influence-ub}
\violet{For any $f: \zo^n \to \zo$,} $\Inf(f) \le \Cert(f)$. 
\end{corollary}

The corollary provides the upper bound on $\phi(0)$, since $\phi(0)$ is just $\Inf(f)$. To lower bound the drop in $\phi$ as we increment $d$,
we need the following two claims. Together, they lower bound the drop in influence of subfunctions if you query the most influential variable
in a function with low query complexity.

\begin{theorem}
[Corollary of Theorem 1.1 from \cite{OSSS05}\footnote{In \cite{OSSS05}, they do not have a factor of $4$ because their function has range $\bits$.}]
    \label{thm:OSSS}
    For any $f:\zo^n \to \zo$,
    \begin{align*}
        \maxinf(f) \geq \frac{4\Var(f)}{\Depth(f)}.
    \end{align*}
\end{theorem}

\begin{proposition}
\label{prop:avg-inf}
For any function $f$ and variable $x_i$, 
\[\Inf(f) - \Inf_i(f) = \lfrac{1}{2}(\Inf(f_{x_i=0}) + \Inf(f_{x_i=1})).\]
\end{proposition}

Finally, in addition to these statements about the behavior of $\phi$,
we also need a condition under which we can be assured
that a leaf is constant.
\begin{fact}[Granularity of variance]
\label{fact:granularity}
Let $f$ be a function with certificate complexity $k$. 
Then if $\Var(f) < 2^{-k} - 2^{-2k}$, it must be the case that $\Var(f) = 0$,
and thus $f$ is constant.
\end{fact}

\begin{proof}
$\Var(f)$ is defined as $\Pr[f(\bx)=0] \cdot \Pr[f(\bx)=1]$.
If $f(x) = 0$ for some $x$, then there must be a 0-certificate for $x$;
i.e. there must be a subcube of codimension $\le k$ for which 
$f$ is constant 0. 
Thus, $\Pr[f(\bx) = 0] \ge 2^{-k}$, and more generally,
if $f$ is not constant then $\min(\Pr[f(\bx)=0], \Pr[f(\bx)=1]) \ge 2^{-k}$. 
The function $p(1-p)$, for $p$ ranging from $2^{-k}$ to $1 - 2^{-k}$,
is minimized at the endpoints of that interval, 
where it takes the value $2^{-k} - 2^{-2k}$.
\end{proof}

\subsection{Proof  of~\Cref{lem:most-leaves-constant}}

With the claims presented in the previous subsection in hand,
we can now prove \Cref{lem:most-leaves-constant}.
A lower bound on the drop in $\phi$ in terms of variable influences in $T_f(d-1)$ follows directly from \Cref{prop:avg-inf}:
\begin{corollary}[Corollary of \Cref{prop:avg-inf}]
\label{cor:avg-inf}
For any function $f$ and depth $d \le n$, 
\[\phi(d) = \phi(d-1) - \E_{\bell \in T_f(d-1 )}[\Inf_{i(\bell)}(f_{\bell})],\]
where $i(\ell)$ is the variable queried at $\ell$ in $T_f(d-1)$.
\end{corollary}

\begin{lemma}[Upper bound on $\phi(d)$]
\label{lem:phi-d-ub}
For a function $f$ with certificate complexity $k$,
and any depth $d \le n$, we have 
\[\phi(d) \le k\cdot \bigg( 1 - \frac{1}{k^3} \bigg)^d\]
\end{lemma}

\begin{proof}
Each depth-$d$ node of $T_f(d)$ maximizes influence
in the corresponding leaf function $f_\ell$ of $T_f(d-1)$.
By \Cref{thm:OSSS}, this means its influence on $f_\ell$ 
is at least $4\Var(f_\ell)/k^2$.
Because $f_v$ has certificate complexity at most $k$,
\Cref{lem:basics} guarantees that $4\Var(f_v)/k^2 \ge \Inf(f_v)/k^3$. 
We now apply \Cref{cor:avg-inf}:
\begin{align*}
    \phi(d) &= \phi(d-1) - \E_{\bell \in T_f(d-1 )}[\Inf_{i(\bell)}(f_{\bell})] \tag{\Cref{cor:avg-inf}} \\
    &\le \phi(d-1) - \E_{\bell \in T_f(d-1 )}[\Inf(f_{\bell})/k^3] \tag{\Cref{thm:OSSS}} \\
    &\le \phi(d-1) - \frac{\phi(d-1)}{k^3} \tag{definition of $\phi$}\\
    &\le \phi(0) \cdot \bigg(1 - \frac{1}{k^3} \bigg)^d. \tag{solution to recurrence relation}
\end{align*}
The lemma follows from the fact that $\phi(0) \le k$, which is a
consequence of \Cref{cor:total-influence-ub}. \end{proof}

\begin{proof}[Proof of \Cref{lem:most-leaves-constant}]
Let $d = k^3(2k + \ln k + \log(1/\eps))$, which is $O(k^4 + k^3\log(1/\eps))$.
We begin with the fact that 
\begin{align*}
    \phi(d) &\le k \cdot \bigg(\bigg(1 - \frac{1}{k^3} \bigg)^{k^3}\bigg)^{2k + \ln k + \log(1/\eps)} \\
    &\le k \cdot e^{-(2k + \ln k + \log(1/\eps))} \\
    &< \eps \cdot 2^{-2k}.
\end{align*}

By Markov's inequality, at least a $1 - \eps$ fraction
of the leaf functions must have influence at most $2^{-2k}$. 
By \Cref{lem:basics}, this implies that their variance is also at most $2^{-2k}$.
Since each leaf function has certificate complexity at most $k$, 
from \Cref{fact:granularity} we may infer that these leaves are constant,
which concludes the proof.
\end{proof}

\subsection{A robust version of~\Cref{lem:most-leaves-constant}}
\label{sec:robust} 

Algorithmically, we need not build the influence-maximizing tree
in order for the potential function argument to go through.
Suppose we estimate influences to within a factor of two instead.
Then the influence of the queried variable is at least $2\Var(f_v)/k^2$,
and the multiplicative drop in $\phi(d)$ becomes $(1 - 1/2k^3)$.
Thus, to achieve the $1-\eps$ probability guarantee
with 2-approximate influences, it suffices to grow the tree to 
twice the depth. 
We state the ``robust'' version of this statement as a corollary.

\begin{corollary}[Robust version of \Cref{lem:most-leaves-constant}]
\label{cor:robust-most-leaves-constant}
Let $f$ be a function with certificate complexity at most $k$, and $T_f$ be any decision tree of depth $O(k^4 + k^3 \log (1/\eps))$ such that for each internal node $v$, the variable $i(v)$ queried has influence within a factor of 2 of the largest influence in the subfunction $f_v$:
\[\Inf_{i(v)}(f_v) \ge \lfrac{1}{2}\maxinf(f_v).\]
Then, for a leaf $\bell$ drawn uniformly at random from $T_f$, we have 
\[\Prx_{\bell}[\,f_{\bell} \text{ is constant}\,] \ge 1 - \eps.\]
\end{corollary}



\section{Balanced influences and an algorithmic version of~\Cref{cor:robust-most-leaves-constant}}

In this section, we will prove \Cref{lem:certify arbitrary input}, restated below.
\begin{lemma}[Formal version of \Cref{lem:certify arbitrary input}]
    \label{lem:restrict-formal}
    For any $n \in \N$ and $\delta \in (0,1)$, given any $\poly(n)$-sized circuit $f: \zo^n \to \zo$ and an $\NP$ oracle, there is a randomized algorithm running in $\poly(n, \log(1/\delta))$-time that, with probability at least $1 - \delta$, finds a restriction $\pi$ to $O(\Cert(f)^4)$ variables such that $f_\pi$ is a constant function.
\end{lemma}

\begin{figure}[h!]
  \captionsetup{width=.9\linewidth}
\begin{tcolorbox}[colback = white,arc=1mm, boxrule=0.25mm]
\vspace{3pt}

$\textsc{Find-Restriction}(f)$:
\begin{itemize}[align=left]
	\item[\textbf{Given:}] A circuit representation of $f:\zo^n\to\zo$ and an $\NP$ oracle. 
	\item[\textbf{Output:}] A  restriction $\pi$ such that $f_\pi$ is constant.
\end{itemize}
\begin{enumerate} 
\item Initialize $\pi \leftarrow \varnothing$.
\item While $f_\pi$ is not constant (which is checked using the $\NP$ oracle), \{
\begin{enumerate}
    \item Using \Cref{lem:est-bal-inf}, estimate $\balinf_i(f_\pi)$ for each $i \in [n]$ to within $\pm \frac1{4n}$ with success probabilities $\geq 1 - \frac{1}{4n^2}$ and let $i^\star$ be the index with the largest estimate.
    \item Update $\pi$ with a random restriction to $x_{i^\star}$: $\pi \leftarrow \pi \cup \{x_{i^\star} \leftarrow \bb\}$ where $\bb \sim \zo$ is uniform random.
\end{enumerate}
\}
\item Return $\pi$.
\end{enumerate} 
\vspace{1pt} 
\end{tcolorbox}
\caption{Algorithm for finding a restriction that fixes $f$ to a constant.}
\label{fig:cert-any}
\end{figure}

\Cref{lem:restrict-formal} is an easy consequence of \Cref{lemma:restrict-low-success}: The algorithm in \Cref{lem:restrict-formal} simply repeats \textsc{Find-Restriction}$(f)$ $\log(1/\delta)$ times and returns the shortest certificate it outputted.

\begin{lemma}
    \label{lemma:restrict-low-success}
    The algorithm \textsc{Find-Restriction}$(f)$ from \Cref{fig:cert-any}, with probability at least $\frac{1}{2}$, returns a restriction $\pi$ of length at most $O(\Cert(f)^4)$.
\end{lemma}

Before proving \Cref{lemma:restrict-low-success} we discuss why it uses balanced influences rather than influences. As we are applying \Cref{cor:robust-most-leaves-constant}, we wish to determine a variable with influence at least half as large as $\maxinf(f)$. The definition of influence immediately suggests an algorithm for estimating influences: Randomly sample $\bx \sim \zo^n$ and then compute whether $f(\bx) \neq f(\bx^{\oplus i})$. Using $\poly(1/\eps)$ iterations of that procedure, we can estimate influences to \emph{additive} accuracy $\pm \eps$. To guarantee we pick a variable with influence at least half of $\maxinf(f)$, we would want to estimate influences to accuracy in $\pm\, \frac{\maxinf(f)}{4}$. Unfortunately, the naive influence estimator requires $1 / \poly(\maxinf(f))$ queries to do so, which is intractable when $\maxinf$ is too small. 

A first hope is that the bound of $\maxinf(f) \geq 4\Var(f)/\Depth(f)$ from \Cref{thm:OSSS} will ensure that $\maxinf$ is not too small. Unfortunately, in the proof of \Cref{lem:most-leaves-constant}, we can have $\Var$ as small as $2^{-\Cert(f)}$. When $\Var(f)$ is small, almost all inputs are labeled the same way by $f$, so it is unlikely that a random edge $(\bx, \bx^{\oplus i})$ will be labeled differently. \emph{Balanced influences} (\Cref{def:balanced-influences}) correct for this effect by ensuring that the initial point, $\bx$, is equally likely to be labeled $0$ or $1$ by $f$.

If $f$ is already balanced, meaning it is equally likely to output $0$ or $1$, then $\mathcal{D}_{\bal}^{(f)}$ is just the uniform distribution and $\balinf_i(f) = \Inf_i(f)$. Otherwise, as we show in the following Lemma, balanced influences are proportional to influences scaled by variance.
\begin{lemma}[Balanced influences proportional to influences]
    \label{lem:balanced-influences-var}
    For any non-constant function $f: \zo^n \to \zo$ and coordinate $i \in [n]$,
    \begin{equation*}
        \balinf_i(f) = \frac{\Inf_i(f)}{4\Var(f)}.
    \end{equation*}
\end{lemma}
\begin{proof}
    Let $S_i \coloneqq \{x \in \zo^n \mid f(x)\neq f(x^{\oplus i})\}$ be the set of all inputs to $f$ that are sensitive at the $i^{\text{th}}$ coordinate. Then,
    \begin{equation*}
        \Inf_i(f) = \Prx_{\bx \sim \zo^n}\left[\bx \in S_i\right] = \frac{|S_i|}{2^n}.
    \end{equation*}
    Similarly,
    \begin{equation*}
        \balinf_i(f) = \frac{1}{2}\left(\Prx_{\bx \sim f^{-1}(0)}[\bx \in S_i] + \Prx_{\bx \sim f^{-1}(1)}[\bx \in S_i] \right).
    \end{equation*}

    The key observation is that for every $x \in S_i$ it is also true that $x^{\oplus i} \in S_i$, so $f$ classifies exactly half the inputs in $S_i$ as $0$, and the other half as $1$. Therefore,
    \begin{equation*}
        \balinf_i(f) = \frac{1}{2}\left(\frac{\frac{|S_i|}{2}}{|f^{-1}(0)|} + \frac{\frac{|S_i|}{2}}{|f^{-1}(1)|} \right).
    \end{equation*}
    For each $b \in \zo$, we use $p_b$ as shorthand for $\Prx_{\bx \sim \zo^n}[f(\bx) = b]$. Then,
    \begin{align*}
        \balinf_i(f) &= \frac{1}{4}\left(\frac{|S_i|}{2^n p_0} + \frac{|S_i|}{2^n p_1}\right) \\
        &= \frac{|S_i|}{2^n} \cdot \frac{p_1 + p_0}{4p_0p_1}\\
        &= \Inf_i(f) \cdot \frac{1}{4p_0p_1} \tag{$\Inf_i(f) = \frac{|S_i|}{2^n}, p_1 + p_0 = 1$} \\
        &= \frac{\Inf_i(f)}{4\Var(f)}. \tag{$\Var(f) = p_0 p_1$}
    \end{align*}
\end{proof}
The following is a direct corollary of \Cref{lem:balanced-influences-var}, \Cref{thm:OSSS}, and the fact that $\Depth(f) \leq n$ for any $f:\zo^n \to \zo$.
\begin{corollary}
    \label{cor:bal-inf-OSSS}
    For any non-constant function $f: \zo^n \to \zo$, there is a coordinate $i \in [n]$ satisfying
    \begin{equation*}
        \balinf_i(f) \geq \frac{1}{\Depth(f)} \geq \frac{1}{n}.
    \end{equation*}
\end{corollary}

We are now able to prove \Cref{lemma:restrict-low-success} contingent on \Cref{lem:est-bal-inf} stating that balanced influences can be efficiently estimated, which will appear in the next subsection.
\begin{proof}[Proof of \Cref{lemma:restrict-low-success}]
    \textsc{Find-Restriction} makes at most $n^2$ estimates of balanced influence, so with probability at least $\frac{3}{4}$, all those estimates are within $\pm \frac{1}{4n}$ of the true balanced influences. By \Cref{cor:bal-inf-OSSS}, there is always a variable with balanced influence at least $\frac{1}{n}$, so as long as all estimates succeed, \textsc{Find-Restriction} will always choose an $i^\star$ such that $\balinf_{i^\star}(f) \geq \max_i(\balinf_i(f))$. By \Cref{lem:balanced-influences-var}, that also implies that $\inf_{i^\star}(f) \geq \maxinf(f)/2$.

    As a result, \textsc{Find-Restriction} builds a uniformly random path of the tree $T_f$ described in \Cref{cor:robust-most-leaves-constant}. All but $\frac{1}{4}$-fraction of paths in $T_f$ at depth $O(\Cert(f)^4)$ reach a restriction $\pi$ for which $f_{\pi}$ is constant. Hence, by a union bound, the probability \textsc{Find-Restriction} does not terminate with $|\pi| \leq O(\Cert(f)^4)$ is at most $\frac{1}{4} + \frac{1}{4} \leq \frac{1}{2}$.
\end{proof}


\subsection{\violet{Estimating balanced influences}}
\Cref{def:balanced-influences} suggests a procedure for estimating balanced influences: Randomly sample $\bx \sim \mathcal{D}_{\bal}^{(f)}$ and compute whether $f(\bx) \neq f(\bx^{\oplus i})$. Using $\poly(1/\eps)$ iterations of that procedure is sufficient to estimate balanced influences to additive accuracy $\pm \eps$. 

The challenge in sampling from $\mathcal{D}_{\bal}^{(f)}$ is that when $\Var(f)$ is small, the $\mathcal{D}_{\bal}^{(f)}$ is far from the uniform distribution. Here, we utilize an $\NP$ oracle for a \emph{uniform generation} algorithm.\footnote{\cite{JVV86} gave an algorithm for \emph{approximate} uniform generation, which would have also sufficed for our purpose. For simplicity, we cite the more recent work that gives \emph{exact} uniform generation.}
\begin{theorem}[Uniform generation with an $\NP$ oracle, \cite{BGP00}]
    \label{thm:uniform-generation}
    There is an efficient randomized algorithm $\mcA$ which, given a satisfiable poly-sized circuit $f: \zo^n \to \zo$ and $\NP$ oracle, with high probability, outputs a uniform $\bx \sim f^{-1}(1)$. In the failure case, $\mcA$ outputs $\perp$.
\end{theorem}
As an easy corollary, we can sample from $\mathcal{D}_{\bal}^{(f)}$.
\begin{corollary}[Sampling from $\mathcal{D}_{\bal}^{(f)}$]
    \label{cor:sample-balanced}
    There is an efficient randomized algorithm $\mcA$ which, given a poly-sized non-constant circuit $f: \zo^n \to \zo$, with high probability outputs a uniform $\bx \sim \mathcal{D}_{\bal}^{(f)}$. In the failure case, $\mcA$ outputs $\perp$.
\end{corollary}
\begin{proof}
    $\mcA$ first samples a uniform $\bb \sim \zo$. If $\bb$ is $1$, it uses the uniform generation algorithm from \Cref{thm:uniform-generation} on $f$. Otherwise, it uses that uniform generation algorithm on $\lnot f$, which is still a poly-sized circuit.

    By union bound, the failure probability (of outputting $\perp$) of $\mcA$ is still small. When it does not fail, $\mcA$ outputs a uniform sample from $\mathcal{D}_{\bal}^{(f)}$.
\end{proof}

Finally, we note that \Cref{lem:est-bal-inf} is a direct consequence of \Cref{cor:sample-balanced} and \Cref{def:balanced-influences}: The algorithm with failure probability $1 - \delta$ takes $\poly(1/\eps, \log(1/\delta))$ samples $\bx \sim \mathcal{D}_{\bal}^{(f)}$ and counts the fraction of those samples for which $f(\bx) \neq f(\bx^{\oplus i})$.

    


\violet{\subsection{Technical remarks}} 
\label{subsec:technical remarks}

\paragraph{\violet{Comparison} with \cite{BKLT22}.}
    \cite{BKLT22} solved the certification problem for \emph{monotone} $f$, and also computed influences over a certain balanced distribution as a key step. For any non-constant monotone $f: \zo^n \to \zo$, there is guaranteed to be some $p^\star \in (0,1)$, called the \emph{critical probability}, satisfying,
    \begin{equation*}
        \Ex_{\bx \sim \mathcal{D}_{p^\star}}[f(\bx)] = \frac{1}{2}
    \end{equation*}
    where $\bx \sim \mathcal{D}_{p^\star}$ means each $\bx_i$ is independently $1$ with probability $p^\star$ and $0$ otherwise. \cite{BKLT22} are able to show that there is a coordinate with high influence at the critical probability, meaning $\Prx_{\bx \sim \mathcal{D}_{p^\star}}[f(\bx) \neq f(\bx^{\oplus i})]$ is large.

    While the ways in which \cite{BKLT22} and this work use the existence of a coordinate with high influence is quite different, we find it intriguing that both rely on finding such a coordinate on a distribution on which $f$ is balanced. \violet{A natural avenue for future work is to establish a formal connection between the two works.} Is there some definition of ``balancing a distribution for $f$," that encompasses both this work and \cite{BKLT22}'s techniques, and is sufficient for certification?

\paragraph{An alternative approach through approximate counting.}
    Given a poly-sized circuit $f:\zo^n \to \zo$ and coordinate $i \in [n]$, we can construct the poly-sized circuit $g_i(x) \coloneqq \Ind[f(x) \neq f(x^{\oplus i})$, which measures whether $f$ is sensitive in the $i^{\text{th}}$ direction at the input. In the proof of \Cref{lem:balanced-influences-var}, we saw that $\Inf_i(f) = \frac{|g_i^{-1}(1)|}{2^n}$. This gives an alternative approach to estimating the coordinate with largest influence: Approximately count the number of accepting inputs of $g_i$ for each $i \in [n]$ and output the $i$ maximizing that count. This approximate count just needs to be accurate to a constant \emph{multiplicative} accuracy (i.e., $\mathrm{count}_i \in [\frac{3}{4} \cdot |g_i^{-1}(1)|, \frac{5}{4} \cdot |g_i^{-1}(1)|]$).

    Classical results of \cite{sipser83, stockmeyer83} show that approximate counting can be done efficiently and deterministically given a $\Sigma_2$ oracle, and \cite{JVV86} observed that implicit in those results, an $\NP$ oracle is sufficient if we allow for randomized algorithms. Furthermore, they proved an equivalence between (almost) uniform generation and approximate counting, so it's not surprising that our algorithm has two alternative approaches, one through uniform generation and one through approximate counting.

    Given two equivalent approaches, we chose to highlight that through uniform generation as we find it more intuitive. It also illuminates the possible connection, detailed earlier, with \cite{BKLT22}'s algorithm for certifying monotone functions.

\section{Our certification algorithm: \violet{Proof of~\Cref{thm:upper bound intro}}}


In this section, we show how any procedure that takes a function as input and outputs a certificate for an \textit{arbitrary} input can be converted into a procedure that takes a function and a string as inputs and outputs a certificate for the function's value on that string. Using this construction and the certification procedure from \Cref{sec:structural}, we obtain our main certification algorithm.

\begin{figure}[h!]
  \captionsetup{width=.9\linewidth}
\begin{tcolorbox}[colback = white,arc=1mm, boxrule=0.25mm]
\vspace{3pt}

$\textsc{Find-Certificate}(f,x^\star)$:
\begin{itemize}[align=left]
	\item[\textbf{Given:}] A circuit representation of $f:\zo^n\to\zo$ and an input $x^\star\in \zo^n$. 
	\item[\textbf{Output:}] A certificate $S\sse [n]$ for $f$'s value on $x^\star$
\end{itemize}
\begin{enumerate}
    \item Initialize $S\leftarrow \varnothing$
    \item While $S$ is not a certificate for $f$'s value on $x^\star$:
    \begin{enumerate}
        \item $S\leftarrow S\cup \textsc{Restriction}(f_{S\leftarrow x^\star})$
    \end{enumerate}
    \item Return $S$
\end{enumerate}
\vspace{1pt} 
\end{tcolorbox}
\caption{Algorithm for the certification problem using {\sc Restriction} as a subroutine.}
\label{fig:algorithm-to-compute-certificate-of-a-specific-input}
\end{figure}

The procedure $\textsc{Restriction}(f)$ in \Cref{fig:algorithm-to-compute-certificate-of-a-specific-input} is a (possibly randomized) procedure that takes as input a circuit representation of $f:\zo^n\to\zo$ outputs a set of coordinates $S\sse [n]$ such that there is some $u\in \zo^{|S|}$ such that $f_{S\leftarrow u}$ is a constant function (equivalently $S$ is a certificate for $f$'s value on some input). 

The algorithm works by iteratively building a certificate $S$ using $\textsc{Restriction}(\cdot)$ as a subroutine. While $S$ is not yet a certificate for $f$'s value on $x^\star$, the algorithm calls $\textsc{Restriction}(\cdot)$ on the subfunction obtained by restricting $f$ according to $S$ and $x^\star$. It then augments the candidate certificate with the output from $\textsc{Restriction}(\cdot)$. 

\begin{lemma}[Solving the certification problem using $\textsc{Restriction}(\cdot)$]
\label{lem:findcert-from-certificate}
The algorithm in \Cref{fig:algorithm-to-compute-certificate-of-a-specific-input} runs in polynomial-time with access to an $\NP$ oracle and outputs a certificate of size at most $2\cdot\Cert(f)\cdot\gamma(\Cert(f))$ for $f$'s value on $x^\star$ where $\gamma:\N\to\N$ is any nondecreasing function satisfying $|\textsc{Restriction}(f)|\le \gamma(\Cert(f))$. 
\end{lemma}

We write $\Cert_0(f)=\max_{x\in f^{-1}(0)}\{\Cert(f,x)\}$ to denote the $0$-certificate complexity of $f$ and $\Cert_1(f)=\max_{x\in f^{-1}(1)}\{\Cert(f,x)\}$ for the $1$-certificate complexity of $f$. Our proof relies on the following fact which states that the $0$-certificate complexity of a function decreases when restricting the coordinates of a $1$-certificate to any set of values and vice versa for a $1$-certificate. For the proof of this result, see \cite[Theorem 6]{BKLT22}. The main idea underlying the proof is that every $0$-certificate has a nonempty intersection with every $1$-certificate.

\begin{proposition}[See {\cite[Theorem 6]{BKLT22}}]
\label{prop:one-sided-certificate-complexity-decreases}
Let $f:\zo^n\to\zo$, $S\sse[n]$ and $u\in\zo^{|S|}$. If $f_{S\leftarrow u}(x)$ is the constant function for all $x\in \zo^{n-|S|}$, then
$$
\Cert_1(f)+\Cert_0(f)-1\ge \Cert_0(f_{S\leftarrow u})+\Cert_1(f_{S\leftarrow u})
$$
for all $u'\in \zo^{|S|}$.
\end{proposition}


\begin{proof}[Proof of \Cref{lem:findcert-from-certificate}]
We claim that the algorithm depicted in \Cref{fig:algorithm-to-compute-certificate-of-a-specific-input} satisfies the lemma statement. 

Each step of the algorithm can be implemented efficiently using an $\NP$ oracle. In particular, determining if $S\sse [n]$ is a certificate for $f$'s value on $x^\star$ is equivalent to checking if $f_{S\leftarrow x^*}$ is the constant function. This task can be accomplished by restricting the circuit for $f$ to obtain a circuit for $f_{S\leftarrow x^*}$ and checking whether the restricted circuit is satisfiable (or, possibly unsatisfiable depending on whether $f(x^\star)=1$). Moreover, the correctness of the algorithm follows immediately from the condition of the while loop. Therefore, it suffices to bound the runtime.

If $S=\textsc{Restriction}(f)$, then
\begin{align*}
    \Cert_1(f)+\Cert_0(f)-1 &\ge\Cert_1(f_{S\leftarrow u})+\Cert_0(f_{S\leftarrow u})\tag{\Cref{prop:one-sided-certificate-complexity-decreases}}\\
    &>0
\end{align*}
for all $u$ such that $f_{S\leftarrow u}$ is nonconstant. Therefore, each step of the algorithm decreases the sum $\Cert_1(f_{S\leftarrow x^*})+\Cert_0(f_{S\leftarrow x^*})$ by at least $1$ and terminates when this quantity reaches $0$. It follows by induction that the main loop terminates after at most $2\cdot \Cert(f)\ge \Cert_0(f)+\Cert_1(f)$ calls to $\textsc{Restriction}(\cdot)$. Each loop iteration adds
\begin{align*}
    |\textsc{Restriction}(f_{S\leftarrow x^*})|&\le \gamma(\Cert(f_{S\leftarrow x^\star}))\\
    &\le \gamma(\Cert(f))\tag{$\Cert(f_{S\leftarrow x^*})\le \Cert(f)$}
\end{align*}
coordinates to the candidate certificate. Hence, the overall size of the final certificate is bounded by $2\cdot\Cert(f)\cdot\gamma(\Cert(f))$. 
\end{proof}

\paragraph{Proof of \Cref{thm:upper bound intro} using \Cref{lem:findcert-from-certificate,lem:certify arbitrary input}.}{
By \Cref{lem:certify arbitrary input}, there is a randomized polynomial-time algorithm with an $\NP$ oracle that implements $\textsc{Restriction}(f)$. This algorithm has the guarantee that $|\textsc{Restriction}(f)|\le {O}(\Cert(f)^4)$. Therefore, the algorithm in \Cref{fig:algorithm-to-compute-certificate-of-a-specific-input} satisfies \Cref{thm:upper bound intro} by \Cref{lem:findcert-from-certificate}: it runs in polynomial-time with $\NP$ oracle access and w.h.p. (over the randomness of $\textsc{Restriction}(f)$) outputs a certificate of size ${O}(\Cert(f)^5)$ for $f$'s value on a given input.\hfill\qed
}

\section{\violet{Hardness of instance-wise guarantees: Proof of~\Cref{thm:hardness of approximation}}}
\label{sec:lowerbounds}


\Cref{thm:hardness of approximation} follows as a consequence of the following lemma.

\begin{lemma}
\label{lem:reduction-from-mmww-to-certification}
There is a polynomial-time algorithm that takes a co-nondeterministic circuit $C:\zo^n\times\zo^n\to\zo$ accepting a nonempty monotone set and a parameter $m\ge n$ and outputs a circuit representation of a function $f:\zo^{n+m}\to\zo$ which satisfies the following.
\begin{enumerate}
    \item If $C$ accepts an input with $\le k$ ones, then $\Cert(f,(1^n,1^m))\le k$.
    \item If every $x$ that $C$ accepts has $>k'$ ones, then $|S|>\min\{k',m-n-\log^c(m)\}$ for every certificate $S\sse[m+n]$ of $f$'s value on $(1^n,1^m)$
\end{enumerate}
where $k'>k$ and $c>1$ is an absolute constant. 
\end{lemma}

The proof of this lemma involves zero-error bit-fixing dispersers. At a high level, a disperser is a nearly-surjective function that remains nearly-surjective even after restricting all but a small number of the input bits. More specifically, an $\eps$-error disperser with entropy threshold $\lambda$ is an efficiently constructible function $D:\zo^m\to\zo^n$ such that the image of the function after restricting all but $\lambda$ input bits is at least an $(1-\eps)$-fraction of $\zo^n$. For our application, we require $\eps=0$ so that the function remains fully surjective after restricting the input bits. 

\begin{definition} [Zero-error bit-fixing disperser with entropy threshold $\lambda$] \label{def:disperser} 
Consider a function $D: \zo^m \rightarrow \zo^n$. Let $S\subseteq [m]$ be a subset of size $|S|\le m - \lambda$ and let $u = \zo^{|S|}$ be an assignment to that subset. Then $D_{S\leftarrow u}: \zo^{m-|S|} \rightarrow \zo^n$ is the function $D$ with the variables in $S$ fixed according to $u$. We say that $D$ is a zero-error bit-fixing disperser with entropy threshold $\lambda$ if for all such $S$ and $u$, the image of $D_{S\leftarrow u}$ is $\zo^n$. 
\end{definition}

We use the following key result due to \cite{GS12} about explicit zero-error bit-fixing dispersers. 

\begin{theorem} [Explicit constructions of zero-error bit-fixing dispersers {\cite[Theorem 9]{GS12}}] \label{thm:dispersers}
There exists a constant $c>1$ such that for large enough $m$ and $\lambda \geq \log^c m$, there is an explicit zero-error bit-fixing disperser $D: \zo^m \rightarrow \zo ^{\lambda -\log^c m}$ with entropy threshold $\lambda$. 
\end{theorem}

With this theorem in hand, we are able to prove \Cref{lem:reduction-from-mmww-to-certification}.

\begin{proof}[Proof of \Cref{lem:reduction-from-mmww-to-certification}]
Let $C:\zo^n\times \zo^n\to\zo$ be co-nondeterministic circuit that accepts a non-empty monotone set and let $m\ge n$ be a parameter. We define the function $f:\zo^{n}\times \zo^{m}\to\zo$ as
$$
f(x,z)\coloneqq C(x,D(z))
$$
where $D:\zo^{m}\to\zo^n$ is a zero-error bit-fixing disperser as in~\Cref{thm:dispersers}. See \Cref{fig:disperser} for an illustration of the circuit construction for $f$. Note that $D$'s entropy threshold is $\lambda=n+\log ^c(m)$. We show that solving the certification problem for $f$'s value on the input $\vec{1}=(1^n,1^{m})$ satisfies the lemma statement. We first consider the case where $C$ accepts a low Hamming weight input and then consider the case where all accepted strings have large Hamming weight.

\paragraph{$C$ accepts an input with $\le k$ ones.}{
Suppose $C$ accepts an input $x\in\zo^n$ with $\le k$ ones. Then, by the monotonicity of $C$, the set $S=\{i:x_i=1\}\sse [n+m]$ consisting of the indices on which $x$ is one is a certificate of size $\le k$ for $f$'s value on the input $\vec{1}$. 
}

\paragraph{Every $x$ that $C$ accepts has $>k'$ ones.}{
In this case, we show that every certificate $S$ for $f$'s value on $\vec{1}$ has at least $\min\{k',m-n-\log^c(m)\}$ coordinates. Let $S\sse [n+m]$ be an arbitrary certificate for $f$'s value on $\vec{1}$. We partition $S$ into two sets $S=S_n\cup S_m$ where $S_n\sse [n]$ consists of the indices from the first $n$-coordinates of $(1^n,1^m)$ and $S_m\sse [m]$ consists of the indices from the last $m$-coordinates of $(1^n,1^m)$. We split into two cases depending on the size of $S_m$. In the first case, we lower bound $|S|$ by $k'$ and in the second case we lower bound $|S|$ by $m-n-\log^c(m)$. Since the two cases are exhaustive, the lemma follows.
\begin{enumerate}
    \item[$\circ$]\textbf{$S_m$ is small: $|S_m|\le m-(n+\log ^c(m))$.} The image of $D_{S_m\leftarrow \vec{1}}$ is $\zo^n$ because $D$ is a disperser with entropy threshold $\lambda=n+\log^c(m)$. It follows that for every $x\in\zo^{n-|S_n|}$ and every $y\in \zo^n$,
    $$
    C_{S_n\leftarrow \vec{1}}(x,y)=f(\vec{1}).
    $$
    Therefore, $S_n$ alone is a certificate for $f$'s value on $\vec{1}$. But this implies $|S|\ge |S_n|>k'$ by our assumption that every $x$ that $C$ accepts has at least $k'$ ones.
    \item[$\circ$]\textbf{$S_m$ is large: $|S_m|>m-(n+\log ^c(m))$.} In this case, we have $|S|\ge |S_m|>m-(n+\log ^c(m))$ which gives the desired lower bound.
\end{enumerate}
}

Finally, we observe that our algorithm is efficient since a circuit for $f$ can be constructed by adding wires from the output gates of $D$ to the appropriate input gates of $C$. See also \Cref{fig:disperser}.
\end{proof}

\begin{figure}[h!]
    \centering
    \begin{tikzpicture}[tips=proper]
        \node[isosceles triangle,
            draw,
            isosceles triangle apex angle=80,
            rotate=90,
            minimum size=3cm] (T1) at (-3.75,0){};
        \node[isosceles triangle,
            draw,
            isosceles triangle apex angle=80,
            rotate=90,
            minimum size=3cm] (T2) at (3.75,0){};
        \node[trapezium,
            draw,
            trapezium stretches = false,
            minimum width = 4.5cm, 
            minimum height = 1.5cm] (TR) at ([yshift=-1.04cm,xshift=0.52cm]T2.220){};
            
         \draw [-{>[scale=1.5]},
            line join=round,
            decorate, decoration={
                zigzag,
                segment length=10,
                amplitude=2,post=lineto,
                post length=3pt
            }]  (-0.5,0) -- (0.5,0);
            
        \draw[color=black] (T1.center) node [fill=white] {\Large $C$};
        \draw[color=black] (T2.center) node [fill=white] {\Large $C$};
        
        \draw[color=black] (TR.center) node [fill=white] {\Large $D$};

        \draw[{Stealth[scale=0.5]}-{Stealth[scale=0.5]}] ([yshift=-0.3cm]T1.left corner) to node[midway,fill=white!30,scale=1] {\footnotesize $x\in\{0,1\}^n$} ([yshift=-0.3cm,xshift=-0.0cm]T1.west);
        \draw[{Stealth[scale=0.5]}-{Stealth[scale=0.5]}] ([yshift=-0.3cm,xshift=0.0cm]T1.west) to node[midway,fill=white!30,scale=1] {\footnotesize $y\in\{0,1\}^n$} ([yshift=-0.3cm]T1.right corner);
        
        \draw[{Stealth[scale=0.5]}-{Stealth[scale=0.5]}] ([yshift=-0.6cm,xshift=-0.1cm]TR.200) to node[midway,fill=white!30,scale=1] {\footnotesize $z\in\{0,1\}^m$} ([yshift=-0.6cm,xshift=0.1cm]TR.-20);
        
        \draw[{Stealth[scale=0.5]}-{Stealth[scale=0.5]}] ([yshift=-0.3cm]T2.left corner) to node[midway,fill=white] {\footnotesize $x\in\{0,1\}^n$} ([yshift=-0.3cm,xshift=-0.1cm]T2.west);
        
    \end{tikzpicture}
    \caption{Key construction for the proof of \Cref{lem:reduction-from-mmww-to-certification}. The co-nondeterministic circuit $C:\zo^n\times\zo^n\to\zo$ is augmented with a zero-error bit-fixing disperser $D:\zo^m\to\zo^n$. This construction yields a function $f:\zo^n\times\zo^m\to\zo$ whose gapped certificate complexity reflects the gap in the $\mmww$ problem associated with $C$.}
    \label{fig:disperser}
\end{figure}
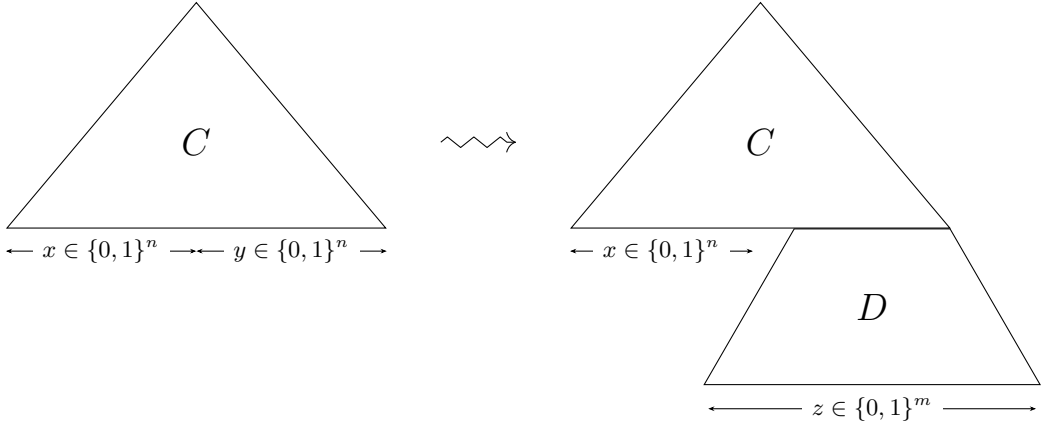

\subsection{Putting it all together: proof of~\Cref{thm:hardness of approximation}}

{ 
Suppose the theorem were false. That is, there is an efficient algorithm that can distinguish between inputs having certificate size $k$ versus requiring size at least $k\cdot n^{1-\eps}$ for some fixed constant $\eps>0$. Then, we will obtain a contradiction to \Cref{thm:umanslb} by solving $\mmww$ with a gap of $n^{1-\eps}$. 

Let $(C,k)$ be an instance of $\mmww$. Let $f:\zo^{n+m}\to\zo$ and $x^\star=\vec{1}$ be the function and input obtained from \Cref{lem:reduction-from-mmww-to-certification} for $m=3n$. Run the certification algorithm on $(f,\vec{1})$ and output YES if the algorithm outputs YES and NO otherwise. To show correctness, we consider the two cases of \Cref{thm:umanslb}: either $C$ accepts an input of length $k$ or every accepted input has length at least $k\cdot n^{1-\eps}$.

\paragraph{$C$ accepts an input with $k$ ones.}{
By the guarantee of \Cref{lem:reduction-from-mmww-to-certification}, there is a certificate of size $\le k$ for $f$'s value on $\vec{1}$. Therefore, our algorithm correctly outputs YES.
}
\paragraph{All strings $C$ accepts have at least $k\cdot n^{1-\eps}$ ones.}{
Using $k'=kn^{1-\eps}$ and $m=3n$, \Cref{lem:reduction-from-mmww-to-certification} guarantees that all certificates $|S|$ for $f$'s value on $\vec{1}$ have size at least $\min\{kn^{1-\eps},2n-\log^c(3n)\}$. We observe
$$
2n-\log^c(3n)> n\ge k\cdot n^{1-\eps}
$$
where the last step follows by our assumption on $C$ (trivially, any input that $C$ accepts can have at most $n$ ones so $n<kn^{1-\eps}$ would be impossible). Therefore,
$$
|S|>\min\{kn^{1-\eps},2n-\log^c(3n)\}=kn^{1-\eps}
$$
for all certificates $S$ of $f$'s value on $\vec{1}$ and therefore our algorithm correctly outputs NO. 
}\hfill\qed
}

\newpage 
\section*{Acknowledgments}

We thank the ITCS reviewers for their useful comments and feedback. 

Guy, Caleb, Carmen, and Li-Yang are supported by NSF awards 1942123, 2211237, and 2224246. Jane is supported by NSF Award 2006664. Caleb is also supported by an NDSEG fellowship.

\bibliography{bibtex}
\bibliographystyle{alpha}

\appendix
\section{Proof of \Cref{lem:n-gapped-cert-is-np-complete}}
\label{appendix:gappedcert-hardness}

The following proof uses dispersers. See \Cref{sec:lowerbounds} for the requisite definitions.

\paragraph{Reduction.}{
We obtain our hardness via a reduction from \textsc{SAT}. Let $\varphi:\zo^n\to\zo$ be a Boolean formula. If $\varphi(1^n)=1$ output ``satisfiable''. Otherwise, let $D:\zo^m\to\zo^n$ be a disperser obtained from \Cref{thm:dispersers} where $m\ge n$ will be chosen later. Construct a circuit for the function $f=\varphi\circ D$ by constructing $D$ and adding wires from the output gates of $D$ to the inputs of $\varphi$. We view $(f,x^\star=1^m)$ as an instance of \textsc{GappedCert}. Output ``unsatisfiable'' if the \textsc{GappedCert} algorithm outputs \textnormal{NO} and otherwise output ``satisfiable''. 
}

\paragraph{Correctness.}{We show that if $\varphi$ is unsatisfiable then our algorithm correctly outputs ``unsatsfiable'' and otherwise we correctly output ``satsfiable''.

\begin{enumerate}
    \item[$\circ$]\textbf{$\varphi$ is unsatisfiable.} If $\varphi$ is unsatisfiable, then $\Cert(\varphi)=0=\Cert(f)$. In this case, \textsc{GappedCert} outputs NO and our algorithm correctly returns ``unsatisfiable''.
    \item[$\circ$]\textbf{$\varphi$ is satisfiable.} Since $\varphi(1^n)=0$, $\varphi$ is a nonconstant function. The entropy threshold for $D$ is $\lambda=n+\log^c(m)$ which implies
\begin{align*}
    \Cert(f,x)=\Cert(\varphi\circ D,x)&\ge m-\lambda\tag{$\varphi$ is nonconstant}\\
    &=m-n-\log^c(m)
\end{align*}
for all $x\in\zo^m$. Therefore, any algorithm for \textsc{GappedCert}$(0,\Psi(m))$ where $m-n-\log^c(m)\ge \Psi(m)$ will output YES, in which case our algorithm correctly outputs ``satsfiable''. Choosing $m=\Theta(n)$ so that $\Psi(m)=\Omega(n)$ completes the proof.\hfill\qed
\end{enumerate}
}

\end{document}